\documentclass[10pt,techreport]{IEEEtran}

\usepackage{amsmath,amsthm, amsfonts}
\usepackage{amssymb}
\usepackage{enumerate}
\usepackage{graphicx}
\usepackage{tikz}
\usepackage{tikz-qtree}
\usepackage{comment}
\usepackage[utf8]{inputenc}

\newtheorem{lemma}{Lemma}
\newtheorem{theorem}{Theorem}

\newtheorem{corollary}{Corollary}

\newcommand{\M}{\mathcal M}

\begin{document}

\title{Memory size bounds of prefix DAGs}
\author{\IEEEauthorblockN{%
   J\'anos Tapolcai, G\'abor R\'etv\'ari, Attila Kőrösi}

 \IEEEauthorblockA{%
MTA-BME Future Internet Research Group,
High-Speed Networks Laboratory (HSN\emph{Lab}),\\ 
Dept. of Telecommunications and Media Informatics,
Budapest University of Technology,\\
   Email: \{tapolcai, retvari,korosi\}@tmit.bme.hu}
}
\maketitle
\begin{abstract}
In this report an entropy bound on the memory size is given for a compression method of leaf-labeled trees. The compression converts the tree into a Directed Acyclic Graph (DAG) by merging isomorphic subtrees.
\end{abstract}

\section{Coupon Collector's Problem with Arbitrary Coupon Probabilities}

Given a set of $C$ coupons, where $\delta=|C|$ denotes the number of coupons. At each draw  $p_o$ denotes the probability for getting coupon $o$ for $o\in C$. We draw $m$ coupons, and let $E$ denote the expected number of different coupons we have obtained. The task is to give an upper bound on  $E$. 


Let $V$ denote the set of coupons we have after $m$ draw. The probability of having coupon $o$ in $V$ is 
\begin{equation} \label{eq_Pnintrie}
P(o\in V)=1-(1-p_o)^{m}
\end{equation}

Thus the expected cardinality of $V$ is 
\begin{multline} \label{sum_alter}
E(|V|)=\sum_{o\in C}E(I(o\in V) = \sum_{o\in C}P(o\in V) =\\
= \sum_{o\in C} (1-(1-p_o)^m) 
\end{multline}

 Let $H_C$ denote the entropy of the coupon distribution
\begin{equation}\label{lemma_entropy}
H_C = \sum_{o\in C} p_o \log_2 \frac{1}{p_o}
\end{equation}

\begin{lemma}\label{couponlemma}
\[E(V) \leq \min\left\{\frac{m}{\log_2(m)}\cdot H_C+3,m,n\right\} \]
\end{lemma}
for $m\geq 3$.
\begin{proof}
Trivially holds that $E \leq m$ and $E\leq n$. Next, let us expand 
\begin{equation}
\sum_{o\in C} (1-(1-p_o)^m) \leq \frac{m}{\log_2(m)} \sum_{o\in C} p_o \log_2 \frac{1}{p_o} +3
\end{equation}
The above inequality holds if the inequality holds for each $o\in C$. Thus next we prove  that 
\begin{equation}\label{eachp}
1-(1-p_o)^m \leq \frac{m}{\log_2(m)} p_o \log_2 \frac{1}{p_o} \quad \textrm{if }p_o< \frac 1e
\end{equation}
holds for $p_o \leq \frac{1}{e}$. Let us assume $m\geq\frac1p$. Note that the right hand size is a monotone increasing function of $m$, when $m>e$. Thus we can substitute $m=\frac1p$ if $\frac1p>e$ in the right hand side and we get
\begin{equation}\label{eachp3}
1-(1-p_o)^m \leq \frac{1/p_o}{\log_2(1/p_o)} p_o \log_2 \frac{1}{p_o} = 1.
\end{equation}
In the rest of the proof we focus on the other case, which is $m<\frac1{p_o}$. Let us define $1>x>0$ as
\begin{equation}\label{def-x}
x:=\log_{\frac1{p_o}} m \enspace .
\end{equation}
After substituting $m=\frac1{p_o^x}$ we have
\begin{multline} 
1-(1-p_o)^{\frac1{p_o^x}} \leq \frac{\frac1{p_o^x}}{\log_2\left({\frac1{p_o^x}}\right)} p_o \log_2 \left(\frac{1}{p_o}\right) =\\= \frac{\frac1{p_o^x}}{x \log_2\left({\frac1{p_o}}\right)} p_o \log_2 \left(\frac{1}{p_o}\right) = \frac{\frac1{p_o^x}}{x} p_o = \frac1{xp_o^{x-1}} \enspace ,
\end{multline}
which can be reordered as
\begin{equation} 
(1-p_o)^{\frac1{p_o^x }}\geq 1 - \frac1{xp_o^{x-1}} \enspace.
\end{equation}
Taking the $p_o^{x-1}>0$ exponent of both sides we get
\begin{equation} \label{reducedeq}
(1-p_o)^{\frac1{p_o}}\geq \left({1 - \frac1{xp_o^{x-1}}}\right)^{p_o^{x-1}}  .
\end{equation}
Note that $x<1$, thus $\frac1x>1$, and we can prove
\begin{equation} 
(1-p_o)^{\frac1{p_o}}\geq \left({1 - \frac1{p_o^{x-1}}}\right)^{p_o^{x-1}}  .
\end{equation}
Bernoulli discovered that $(1-p_o)^{\frac1{p_o}}$ is monotone decreasing function and equals to $\frac1e$ for $p_o \to 0$ \cite{}. Thus the inequality holds if 
\begin{equation} 
p_o \leq \frac1{p_o^{x-1}}  .
\end{equation}
which holds because
\begin{equation} 
p_o^x \leq 1.
\end{equation}
This proves \eqref{eachp} with the assumption of $p_o < \frac{1}{e}$. There are at most $3>\frac1e$ coupons for which \eqref{eachp} cannot be applied, but the expected number of these coupons is still at most $3$.
\end{proof}

\section{Trie-folding}

For IP address lookup a binary trie is used, where each leaf has a label called next hop. To compress the trie we will use trie-folding, which merges the sub-tries with exactly the same structure and next hops labels at each leaf instead of repeating it in the binary trie. after the process the trie is transformed into a DAG. See an example below
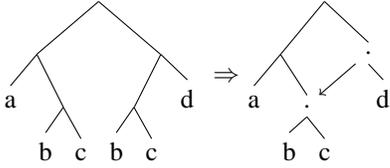
\begin{center}\begin{tikzpicture}[level distance=20pt]
\Tree [  [ a [ b c ] ] [ [ b c ] d ] ]
 \node at (1.7,-1) {$\Rightarrow$};
\begin{scope}[shift={(3cm,0)}]
\Tree [ [ a [.\node(l2){.}; b c ] ] [.\node(l1){.};  \edge[white]; {} d ] ]
\end{scope}
\draw[->] (l1) --  (l2);
\end{tikzpicture}\end{center}

We evaluate the efficiency of the trie-folding methods on a randomly generated trie, where the next hops follow a given distribution. The randomly generated trie is denoted by $T=(V_T,E_T)$ and has the following properties
\begin{description}
\item[$h$] is the height bound of the trie, typically $24$ in IPv4. 
\item[$\delta$] is the set of next hops.
\item[$p_i$] is the probability that an IP address is forwarded to next hop $i\in N$.
\end{description}

Let $V^j_T$ denote the set of nodes in $T$ at the $j$ level for $1\leq j\leq h$. The level of a node is $h$ minus the hop count of the  path to the root node. Thus the root node has level $h$. At the $j$ the level there are $2^{h-j}$ nodes, formally  $|V^j_T|=2^{h-j}$. Each node at the $j$-th level has $2^{j+1}$ child node, and eventually $2^j$ leafs each of which is assigned with a next hop.
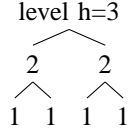
\begin{center}\begin{tikzpicture}[level distance=20pt]
\Tree [.{level h=3}  [.2 1 1  ] [.2 1 1 ] ]
 \end{tikzpicture}\end{center}


The DAG resulted by the trie-folding method is denoted by $D=(V_D,E_D)$, and $V^j_D$ denotes the set of nodes in $D$ at the $j$ level for $1\leq j\leq h$.

\begin{lemma}\label{lemma:vdj-ub}
The expected  number of \textbf{Nodes} at the $j$-th level in a DAG resulted by trie-folding of a randomly generated trie with height $h$ and next hop distribution $p_1,\dots,p_N$  is at most
\begin{equation}\label{eq:vjd-ub}
E(|V^j_D|) \leq \min\left\{\frac{H_O}{h-j} 2^{h} + 3, 2^{h-j}, \delta^{2^j} \right\} 
\end{equation}
where $H_O$ denotes the entropy of the next hops
\begin{equation}\label{lemma_entropy}
H_O = \sum_{o\in N} p_o \log_2 \frac{1}{p_o}.
\end{equation}
\end{lemma}
\begin{proof}
We treat the problem as a coupon collection problem, where each coupon is a subtree with $j$ height and $2^j$ next hops on leafs. In other words each coupon is a string with length $2^j$ on alphabet $\delta$, and we draw $m=2^{h-j}$ coupons. Note that there are $C=\delta^{2^j}$ different coupons. Lemma \ref{couponlemma} gives an upper bound on the number of different coupons, which are the subtrees in this case. Thus we have $|V^j_D| \leq 2^{h-j}$, $|V^j_D| \leq \delta^{j}$ and
\begin{multline}
E(|V^j_D|) \leq \frac{2^{h-j}}{\log_2(2^{h-j})} H_C +3 = \frac{2^{h-j}}{\log_2(2^{h-j})} H_O 2^j+3 \\=\frac{2^h}{h-j} H_O+3,
\end{multline}  
where  $H_C=H_O 2^j$ is the entropy of a $2^j$ long string made of next hops. /* we need to find a reference or add a lemma proving it */
\end{proof}

Let $k^*$ be the row where the bounds take the maximum value for all $j=1,\dots,h$. See also Figure \ref{fig:DAG-shape} as an illustration of the bounds on the width of the DAG given by the above lemma. Such $k^*$ clearly exists, because the bounds by Lemma \ref{lemma:vdj-ub} are decreasing function of $j$ until $2^h-j$ holds, while both $\frac{H_O}{h-j} 2^{h} + 3$ and $\delta^{2^j}$ are monotone increasing functions of $j$.  

We store each pointer for a node in $h-k^*$ bits. Since each node has two child nodes, it can be stored in $2h-2k^*$ bits. At level $k^*$ the bound is 
\begin{multline} 
E(|V^{j}_D|) \leq E(|V^{k^*}_D|) \leq \frac{H_O}{h-k^*} 2^{h} + 3 \leq 2^{h-k^*} \\  j=1,\dots, h
\end{multline}
As each node is stored in $2h-2k^*$ bits we have the following corollary on the width of the DAG.

\begin{corollary}\label{corollary-dag-width}
The expected  number of \textbf{bits} to store the nodes at any level $j=1,\dots, h$ in the DAG resulted by trie-folding of a randomly generated trie with height $h$ and next hop distribution $p_1,\dots,p_N$  is at most $$\M=2 H_O 2^{h} + 6h.$$
\end{corollary}

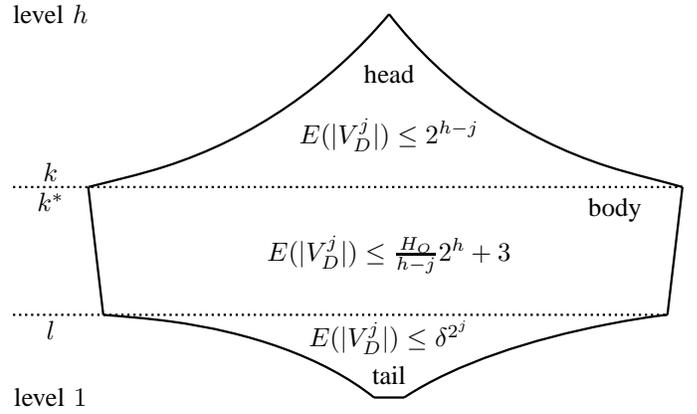
\begin{figure}
\begin{tikzpicture}[x=1.00mm, y=1.00mm, inner xsep=0pt, inner ysep=0pt, outer xsep=0pt, outer ysep=0pt]
\definecolor{L}{rgb}{0,0,0}
\path[line width=0.30mm, draw=L] (70,90) .. controls (62.02,80.52) and (51.82,73.46) .. (40.27,69.75) .. controls (36.78,68.63) and (33.21,67.83) .. (30,67);
\path[line width=0.30mm, draw=L] (70,90) .. controls (77.30,80.52) and (87.50,73.46) .. (99.05,69.75) .. controls (102.54,68.63) and (106.11,67.83) .. (109,67);
\path[line width=0.30mm, draw=L] (30,67) -- (32,50);
\path[line width=0.30mm, draw=L] (68,39) .. controls (61.23,44.41) and (51.30,47.72) .. (41.00,49.11) .. controls (37.70,49.56) and (34.38,49.80) .. (32,50);
\path[line width=0.30mm, draw=L] (68,39) -- (72,39);
\path[line width=0.30mm, draw=L] (72,39) .. controls (79.00,43.60) and (88.89,46.75) .. (99.05,48.74) .. controls (102.11,49.34) and (105.18,49.83) .. (107,50);
\path[line width=0.30mm, draw=L] (109,67) -- (107,50);
\path[line width=0.30mm, draw=L, dash pattern=on 0.30mm off 0.50mm] (20,50) -- (107,50);
\path[line width=0.30mm, draw=L, dash pattern=on 0.30mm off 0.50mm] (20,67) -- (109,67);
\draw(25,48) node{$l$};
\draw(25,69) node{$k$};
\draw(25,65) node{$k^*$};
\draw(25,90) node{level $h$};
\draw(25,39) node{level $1$};
\draw(70,74) node{$E(|V^j_D|) \leq 2^{h-j}$};
\draw(70,58) node{$E(|V^j_D|) \leq \frac{H_O}{h-j} 2^{h} + 3$};
\draw(70,47) node{$E(|V^j_D|) \leq \delta^{2^j}$};
\draw(70,82) node{head};
\draw(100,64) node{body};
\draw(70,42) node{tail};
\end{tikzpicture}%
\caption{\label{fig:DAG-shape} The shape of the bounds on the DAG.}
\end{figure}

Based on this we have the following theorem on the size of the DAG.
\begin{theorem}\label{thm-dag-size-1}
The expected  number of \textbf{bits} to store the nodes in the DAG resulted by trie-folding of a randomly generated trie with height $h$ and next hop distribution $p_1,\dots,p_N$  is at most $$2 h H_O 2^{h} + 6h^2.$$
\end{theorem}

The lower bound above theorem can be further improved if $H_O\geq \frac{h}{2^h}$.
Let $k$ be the smallest level where $\frac{H_O}{h-k} 2^{h} + 3$ is larger than $2^{h-k}$. Note that, $k^*<k$.  The value of $k$ is 
\begin{equation}\label{k-def}
 k > \lceil{ \log_2 \left( \frac{h}{H_O} \right) }\rceil,
\end{equation}
because
\begin{multline} \label{k-bounds}
2^{h-k}<2^{h-\lceil{\log_2 \left( \frac{h}{H_O} \right) }\rceil} \leq 2^{h-\log_2 \left( \frac{h}{H_O} \right) } =\\= 2^h\frac{H_O}{h}  < \frac{H_O}{h-k} 2^{h} + 3
\end{multline}

Note that, $\log_2 \left( \frac{h}{H_O} \right)\leq h$ when $H_O\geq \frac{h}{2^h}$.

To count the total space needed to store the DAG we divide it into two parts (see also Figure \ref{fig:DAG-shape})
\begin{description}
\item[\emph{head}] for levels $h,\dots,k$, 
\item[\emph{body}] for levels $k-1,\dots,1$.
\end{description}

First we estimate the size of head and use the bound $2^{h-j}$ from \eqref{eq:vjd-ub}. The expected number of bits needed for the DAG at level $j=k,\dots, h$ is
\begin{equation*}\label{size-head-node}
\sum_{j=k}^{h}  E(|V^j_D|) \leq \sum_{j=k}^{h}  2^{h-j} = \sum_{j=0}^{h-k}  2^{j} = 2^{h-k+1} < 2 \frac{H_O}{h-k} 2^{h} + 6 
\end{equation*}
where the last inequality comes from \eqref{k-bounds}. After multiplying with $2h-2k^*$ bits for each node we have 
\begin{multline}\label{size-head-bit}
(2h-2k^*) 2\left( \frac{H_O}{h-k} 2^{h} + 6 \right)<4(h-k) \frac{H_O}{h-k} 2^{h} + 12h =\\= 4 H_O 2^{h} + 12h= 2\M
\end{multline}

For the size of body we use Corollary \ref{corollary-dag-width}.
\begin{multline}\label{size-tail-node}
 \sum_{j=l+1}^{k-1} (2h-2k-2) E(|V^j_D|) \leq (k-1) \M =\\=
\left( \lceil{ \log_2 \left( \frac{h}{H_O} \right) }\rceil -1 \right) \M <
\log_2 \left( \frac{h}{H_O} \right) \M 
\end{multline}
Finally, summing up with \eqref{size-head-bit} we get the following bound.
\begin{theorem}\label{thm-dag-size-2}
The expected  number of \textbf{bits} to store the nodes in the DAG resulted by trie-folding of a randomly generated trie with height $h$ and next hop distribution $p_1,\dots,p_N$  is at most $$\left({ 2+\log_2 (h) - \log_2 H_O }\right) \left( 2 H_O 2^{h} + 6h \right),$$
when $H_O\geq \frac{h}{2^h}$.
\end{theorem}

Finally we further improve the lower bound above theorem when $\delta$ is a finite number.
Let $l$ be the largest level where $\delta^{2^l}$ is smaller than $\frac{H_O}{h-l} 2^{h} + 3$. The value of $l$ is 
\begin{equation}\label{l-def}
l <  \lfloor{ \log_2\left(\frac{h-\log_2 \left( \frac{h}{H_O} \right)}{\log_2(\delta)} \right) }\rfloor,
\end{equation}
because 
\begin{multline} \label{l-bounds}
\delta^{2^l}< \delta^{ 2^{\lfloor{\log_2\left(\frac{h-\log_2 \left( \frac{h}{H_O} \right)}{\log_2(\delta)} \right)}\rfloor}}\leq \delta^{ 2^{\log_2\left(\frac{h-\log_2 \left( \frac{h}{H_O} \right)}{\log_2(\delta)} \right)}} =\\
\delta^{ \frac{h-\log_2 \left( \frac{h}{H_O} \right)}{\log_2(\delta)} }=
 \delta^{\frac{\log_2 \left(2^h\right) + \log_2 \left( \frac{H_O}{h} \right)}{\log_2(\delta)}  }= \delta^{\frac{\log_2 \left( \frac{H_O}{h} 2^h\right)}{\log_2(\delta)}  } \\=\delta^{\log_N \left( \frac{H_O}{h} 2^h\right)  } = \frac{H_O}{h} 2^h < \frac{H_O}{h-l} 2^{h} + 3
\end{multline}

Note that $k\geq l+1$, because of the floor an ceiling function and
\[ \log_2 \left( \frac{h}{H_O} \right) > \log_2\left(\frac{h-\log_2 \left( \frac{h}{H_O} \right)}{\log_2(\delta)} \right)\]
and taking both side on power 2 we have
\[  \frac{h}{H_O} > \frac{h-\log_2 \left( \frac{h}{H_O} \right)}{\log_2(\delta)}\]
Note that $H_O\leq \log_2(\delta)$, thus 
\[  h > h-\log_2 \left( \frac{h}{H_O} \right),\]
which always holds.

To count the total space needed to store the DAG we divide it into three parts (see also Figure \ref{fig:DAG-shape})
\begin{description}
\item[\emph{head}] for levels $h,\dots,k$, 
\item[\emph{body}] for levels $k-1,\dots,l+1$,
\item[\emph{tail}] for levels $l,\dots,1$.
\end{description}

To estimate the size of head we use \eqref{size-head-bit}. For the size of the tail we use the bound $\delta^{2^j}$ from \eqref{eq:vjd-ub}.  We have
\begin{multline}\label{size-tail-node}
\sum_{j=1}^{l}  E(|V^j_D|) \leq \sum_{j=1}^{l}  \delta^{2^j} < \sum_{i=1}^{2^l}  \delta^{i} = N \sum_{i=0}^{2^{l}-1}  \delta^{i} =\\= N\frac{\delta^{2^{l}}-1}{\delta-1} < N\frac{\delta^{2^{l}}}{\delta-1}= \frac{\delta}{\delta-1} \left(\frac{H_O}{h-l} 2^{h} + 3\right)
\end{multline}
where the last inequality comes from \eqref{l-bounds}. After multiplying with $2h-2k^*$ bits for each node we have 
\begin{multline}\label{tail-head-bit}
(2h-2k^*) \frac{\delta}{\delta-1} \left(\frac{H_O}{h-l} 2^{h} + 3\right)<\\<\frac{\delta}{\delta-1} \left(2(h-k) \frac{H_O}{h-k} 2^{h} + 12h\right) =  \frac{\delta}{\delta-1} \M
\end{multline}

For the size of body we use Corollary \ref{corollary-dag-width}.
\begin{multline}\label{size-tail-node}
 \sum_{j=l+1}^{k-1} (2h-2k-2) E(|V^j_D|) \leq (k-l-1) \M =\\=
\left( \log_2 \left( \frac{h}{H_O} \right) -  \log_2 \left({\frac{h-\log_2 \left( \frac{h}{H_O} \right)}{\log_2(\delta)}}\right) -1 \right) \M =\\
\left( \log_2 \left( \frac{h}{H_O\left(h-\log_2 \left( \frac{h}{H_O} \right)\right)} \right) + \log\log_2(\delta) -1 \right) \M 
\end{multline}
Finally, summing up with \eqref{tail-head-bit} and \eqref{size-head-bit} we get the following bound.

\begin{theorem}
The expected  number of bits to store the nodes in the DAG resulted by trie-folding of a randomly generated trie with height $h$ and next hop distribution $p_1,\dots,p_N$  is at most
\begin{multline}
\left(1+ \log_2 \left( \frac{h}{h-\log_2 (h) + \log_2 H_O} \right) - \log_2 H_O + \right.\\\left. \log\log_2(\delta) + \frac{\delta}{\delta-1} \right) \left(2 H_O 2^{h} + 6h\right)
\end{multline}
when $H_O\geq \frac{h}{2^h}$ and $\delta$ is a finite number.
\end{theorem}

Note that the above bound asymptotically leads to
\[
\left(1 - \log_2 H_O + \log\log_2(\delta) + \frac{\delta}{\delta-1} \right) 2 H_O 
\]
bits for each leaf when $h\rightarrow \infty$. For $\delta=2$ it is $H_O<1$ and
\[
\left(6 - 2\log_2 H_O\right) H_O. 
\]

\end{document}